\documentclass[12pt]{article}
\usepackage[a4paper,margin=2cm]{geometry}
\usepackage{amsmath,amssymb,amsthm,mathtools,bm}
\usepackage{booktabs,array,multirow,longtable}
\usepackage{graphicx}
\usepackage{tikz}
\usepackage{hyperref}
\usepackage{xcolor}
\usepackage{enumitem}
\usepackage{caption}
\usepackage{subcaption}
\usepackage{float}

\usetikzlibrary{arrows.meta,calc,positioning,decorations.pathreplacing}

\hypersetup{colorlinks=true,linkcolor=blue,citecolor=blue,urlcolor=blue}

\title{Financial Relativity: An Information-Geometric Interpretation of Asset Pricing}
\author{Lin Li}
\date{}

\theoremstyle{plain}
\newtheorem{theorem}{Theorem}[section]

\newtheorem{proposition}[theorem]{Proposition}
\newtheorem{definition}{Definition}[section]
\theoremstyle{remark}

\begin{document}
\maketitle

\begin{center}
\vspace{1mm}
\small{Department of Finance, Business School, East China University of Science and Technology}
\end{center}

\renewcommand{\refname}{{Reference}}
\renewcommand{\tablename}{Table}
\renewcommand{\figurename}{{\bf Figure}}
\vspace{1mm}

\vspace{0.2cm}
\begin{center}
\parbox{\textwidth}{\setlength{\parindent}{0em}
\vspace{3.2mm}

\normalsize{\textbf{Abstract}: {\it
Classical asset pricing relies on the risk-neutral measure $Q$ for valuation, yet its economic meaning is typically grounded in a physical measure $P$. This creates an implicit asymmetry: pricing is governed by $Q$, while interpretation is anchored in $P$. As a result, the same pricing relation may admit multiple, and sometimes conflicting, explanations depending on how the link between $P$ and $Q$ is specified, making it difficult to provide a unified account of asset pricing within a single conceptual framework.
This paper proposes an alternative perspective based on information geometry, termed \emph{financial relativity}. The central principle is a relativity of probabilistic reference frames: $P$ and $Q$ possess no intrinsic hierarchy, but instead correspond to geometric structures induced by different informational constraints. Terminal structural information determines how probability geometry is shaped, and this geometry in turn governs how information is expressed in prices.
Within this framework, the risk-neutral measure is reinterpreted as a posterior probability geometry shaped by structural constraints. Asset prices are characterized as geometric projections of terminal payoffs onto information subspaces, and price dynamics emerge as the progressive manifestation of structural information under evolving geometry. We develop discrete and continuous financial field equations to describe the formation of probability geometry, and derive geodesic price dynamics in which volatility is endogenously determined by posterior uncertainty.
This framework provides a unified explanation for price fluctuations, event-driven dynamics, and risk premia, and yields testable implications, including structural links between volatility and posterior variance, as well as measures of the informational efficiency of prices. By integrating structural information, probability measures, and price dynamics within a unified geometric framework, the paper offers a coherent, extensible, and empirically tractable reinterpretation of asset pricing theory.}
 \vspace{3mm}

}}
\vspace{2mm}
\end{center}

\pagebreak

\section{Introduction}

Modern asset pricing theory has achieved remarkable formal success. Whether in classical no-arbitrage pricing or in continuous-time risk-neutral dynamics, almost all core results can ultimately be expressed in a unified form: prices equal the discounted conditional expectation of future payoffs. Since the seminal work of Black and Scholes \cite{BlackScholes1973} and Merton \cite{Merton1973}, this representation has found its most influential application in derivative pricing. Subsequently, the theory of no-arbitrage and equivalent martingale measures provided a rigorous probabilistic foundation for this formulation \cite{HarrisonKreps1979,HarrisonPliska1981}. The Cox–Ingersoll–Ross framework \cite{CoxIngersollRoss1985} further demonstrated that the risk-neutral probability is not merely a technical device in option pricing, but a general structural feature of modern continuous-time asset pricing. This leads to a clear and standard formulation: as long as there is no arbitrage, there exists a risk-neutral measure $Q$ under which asset prices can be written as discounted conditional expectations of future cash flows \cite{DelbaenSchachermayer2006}.

However, precisely within this highly successful theoretical structure, a long-standing and unresolved question remains: what is the economic meaning of $Q$?

In the classical no-arbitrage framework, $Q$ is often treated merely as a mathematical object arising from an existence theorem. In consumption-based or representative-agent models, it is typically interpreted as a transformation of the physical probability $P$ under marginal utility weights \cite{Lucas1978,HansenJagannathan1991,Cochrane2005}. Because the reciprocal of the Radon–Nikodým derivative corresponds exactly to the stochastic discount factor (or pricing kernel), this interpretation ultimately links the transformation from $P$ to $Q$ to the economic meaning of the stochastic discount factor. A substantial body of modern asset pricing theory and empirical work has indeed been developed around the properties of the stochastic discount factor \cite{HansenJagannathan1991,Campbell2000,Cochrane2005}. In this sense, $Q$ is no longer merely a technical construct, but becomes one of the core probabilistic structures through which prices and uncertainty are connected.

Strictly speaking, even the behavioral finance literature does not fundamentally depart from this framework. Rather than abandoning $Q$-based pricing, it reshapes the structure of the stochastic discount factor through mechanisms such as sentiment, loss aversion, framing effects, probability weighting, or heterogeneous beliefs \cite{Shefrin2005,Shefrin2008,BarberisHuangSantos2001,BarberisHuang2008,PolkovnichenkoZhao2013}. In this sense, behavioral deviations are not a paradigm opposing $Q$-based pricing, but can instead be understood as selecting a different $Q$ modified by cognitive or psychological weighting. Shefrin explicitly argues that the most promising direction in behavioral asset pricing is not to abandon the stochastic discount factor, but to incorporate sentiment, heterogeneous beliefs, and behavioral biases into a unified pricing kernel framework \cite{Shefrin2008}. Polkovnichenko and Zhao \cite{PolkovnichenkoZhao2013} further show that even in empirical pricing kernels extracted from option markets, probability weighting can enter the risk-neutral distribution through its effect on the shape of the pricing kernel. Thus, the issue is not whether $Q$ exists, but how different versions of $Q$ should be understood in terms of their status, origin, and mutual relationships.

Despite its depth, this line of interpretation remains unsatisfactory in several respects. First, it implicitly assumes that the market—an aggregate system composed of heterogeneous agents—can be represented by a single entity with stable preferences and beliefs (even if those beliefs are biased), capable of making consistent decisions. Second, it presumes that a prior “more fundamental” probability $P$ exists and is captured by the market, while $Q$ is treated as a derivative object shaped by preferences and cognition. Third, in empirical applications, preference structures are typically difficult to observe directly, and the identification of risk aversion parameters and pricing kernels often relies on strong assumptions. In particular, following Ross’s Recovery Theorem \cite{Ross2015Recovery}, which reignited the question of whether $P$ can be recovered from $Q$, it becomes even more important to ask which probability measure should be regarded as operationally relevant for the market. Subsequent research has shown that recovering $P$ requires additional structural assumptions that are far from trivial \cite{BorovickaHansenScheinkman2016,BakshiChabiYoGao2018}. In other words, while the Recovery problem is framed as reconstructing the “true” probability from the risk-neutral measure, the deeper question is why $P$ should be regarded as more fundamental than $Q$ in the first place. This priority may reflect a theoretical convention rather than a necessary property of the market.

This paper approaches the problem from a different perspective. It does not challenge the mathematical validity of no-arbitrage pricing, nor does it deny the importance of behavioral biases, risk aversion, or learning mechanisms at the individual level. Rather, it questions whether, for a market composed of interacting agents, trading systems, institutional constraints, and financial technologies, it remains appropriate to interpret it as a “super-agent” characterized by stable preferences. This is not merely a semantic issue. As Kirman \cite{Kirman1992} has argued, the mapping from heterogeneous individuals to aggregate demand does not preserve individual preference structures. Moreover, in asset price dynamics, heterogeneous beliefs, incomplete information, and evolutionary feedback can generate behaviors far more complex than those implied by individual optimization \cite{BrockHommes1998,PastorVeronesi2009}. If such emergent mechanisms partially wash out individual-level preference structures, then interpreting $Q$ simply as a preference-weighted transformation of $P$ is no longer the only, nor necessarily the most natural, explanation.

At the same time, the standard interpretation implicitly relies on an epistemic premise that is rarely made explicit: once a distinction between $P$ and $Q$ is introduced, $P$ is often taken, without further justification, as the “true probability,” while $Q$ is treated as a technically convenient but non-true probability for pricing purposes. Yet asset pricing theory itself does not endow $P$ with such ontological priority. In a finite state space, what is called $P$ is nothing more than a prior assignment of equal weights to primitive states; its special status lies first in its maximal symmetry, not in any inherent proximity to the real world. In other words, the importance of $P$ does not stem from being logically “more true,” but from providing the least-committed and most symmetric probabilistic reference frame in the absence of further information. Once viewed in this way, the relationship between $P$ and $Q$ is no longer simply one between a “true probability” and a “distorted probability,” but rather a distinction between two probabilistic reference frames associated with different informational structures. The notion of a reference frame here is not metaphorical: it refers to the way the market organizes the terminal state space and assigns relative weights across states as a system of probabilistic coordinates. A simple example illustrates this point. Suppose an event has only two terminal states, “success” and “failure.” In the absence of additional information, assigning probability $1/2$ to each state is the natural symmetric choice. However, as signals such as approval progress, order flow, or default indicators are gradually revealed, maintaining the $1/2$–$1/2$ description becomes untenable, and the market must adopt a non-uniform probability structure to reflect the relative likelihood of the two states. What changes in this process is not necessarily preferences, but the probabilistic reference frame under which the market operates. In this sense, traders do not face a fully known and numerically fixed objective distribution; rather, they operate within a state-space structure that must be continually inferred and updated as information is revealed. Even if one grants that probabilities may in some sense exist objectively, the market cannot know their exact values ex ante. What traders and markets can do is to update their assessment of terminal states based on observable revelation processes. As a result, phenomena that are often directly attributed to “preference distortions” may, at a more fundamental level, be reinterpreted as updates and reconfigurations of the probability structure itself. If this organization of relative state weights is viewed as a geometry, then what market pricing relies on is not merely a set of probability numbers, but a probability geometry that evolves with information revelation.

It is precisely at this point that the present paper seeks to revive the notion of “relativity.” In physics, general relativity is called a theory of relativity not simply because it distinguishes reference frames, but because it uses geometry to express the covariance of physical laws across a broader class of frames, thereby removing the a priori privilege of particular frames. The pair $(P,Q)$ in asset pricing can likewise be interpreted as two different probabilistic reference frames. If probability is understood as a tool through which the market quantifies uncertainty, performs calculations, and updates beliefs, then there is no compelling reason to grant $P$ a higher ontological status than $Q$ at the outset. A more natural view is that different probability geometries correspond to different informational structures. When terminal structure is completely indistinguishable, the least-committed reference frame—by Jaynes’s maximum entropy principle \cite{Jaynes1957,Jaynes1982}—is the flat geometry with equal weights on primitive states. When terminal structure becomes partially revealed, the same maximum entropy logic leads to a non-uniform posterior geometry as the appropriate reference frame. Under this view, the difference between $P$ and $Q$ is not primarily a difference in “truth,” but a difference in consistency with the current revelation structure.

Building on this perspective, the paper raises a central question: \emph{if the risk-neutral probability $Q$ is reinterpreted as a posterior probability geometry formed by the market under the current revelation structure, rather than as a preference-weighted instrumental measure, what new organizational structure would emerge for asset pricing theory?} The answer proposed here is that, once this reinterpretation is adopted, the Fundamental Theorem of Asset Pricing, conditional expectation projections, event-driven price dynamics, volatility patterns, and even the Recovery problem can all be understood within a unified geometric language. More specifically, we advance the following proposition: \emph{terminal structural information determines how probability geometry is curved, and probability geometry determines how information is manifested in prices.} This statement closely parallels the canonical formulation in general relativity—\emph{matter tells spacetime how to curve, and spacetime tells matter how to move} (see \cite{MisnerThorneWheeler1973} and related expositions). It is in this precise sense that we refer to the framework as \emph{financial relativity}: rather than mechanically importing concepts from physics, it seeks to establish a coherent theoretical structure linking the equivalence of probabilistic reference frames, the formation of geometric structure, and the covariant expression of price dynamics.

This proposition intersects with several strands of the existing literature, while also differing from them in essential ways. First, relative to the foundational no-arbitrage literature 
, the present paper does not alter the mathematical form of risk-neutral pricing, but reinterprets its conceptual status: $Q$ is no longer merely a measure whose existence suffices for pricing, but is elevated to the role of the probability geometry actually employed by the market. Second, in contrast to preference-based, ambiguity-aversion, or robust control approaches \cite{EpsteinSchneider2010,HansenSargent2001}, we do not derive risk prices from utility, but instead locate the origin of geometry in terminal structure and revelation mechanisms. Third, the framework is closely related to the literature on information-based asset pricing and learning, which emphasizes that prices are formed under incomplete information \cite{GrossmanStiglitz1980, Kyle1985,Back1992, DuffieLando2001}, and \cite{BrodyHughstonMacrina2008, PastorVeronesi2009, Veronesi1999, BrennanXia2001}. However, unlike these contributions, the objective here is not to analyze price behavior under a given information structure, but to propose a unifying organizational principle: discrete no-arbitrage pricing, continuous filtering, risk-neutral propagation, and event-driven volatility are all recast within a three-layer framework of “information–geometry–motion.” Recent work in information-theoretic and entropy-based approaches has also begun to emphasize the recovery of pricing kernels and probability measures from prices \cite{GhoshJulliardTaylor2025}, which resonates with the direction pursued here. To the best of our knowledge, however, no existing work systematically reformulates the Fundamental Theorem of Asset Pricing as a relativity-style theory of probability geometry in the manner proposed in this paper.

The contributions of this paper follow a unified logical structure. First, we conceptually reinterpret the risk-neutral probability. Rather than viewing $Q$ as a transformation of the physical probability $P$ under representative preferences, we interpret it as a posterior probability geometry formed by the market under given informational constraints. This reinterpretation does not deny the importance of preferences at the individual level; rather, it shifts the core of market-level pricing from a “preference–probability” language to an “information–geometry” language, thereby alleviating the tension between no-arbitrage pricing and its economic interpretation. Second, we propose and formalize the basic principles of financial relativity. We show that the statement “matter tells spacetime how to curve, and spacetime tells matter how to move” has a precise analogue in asset pricing. In discrete state spaces, this principle corresponds to a field equation and geometric potential induced by terminal structural sources; in continuous time, it corresponds to a stochastic evolution equation for posterior densities under ongoing informational constraints. In this way, discrete no-arbitrage pricing, conditional expectation projections, continuous filtering, and inertial price propagation are unified within a single theoretical framework. At the same time, we develop a structure that connects discrete and continuous settings. In finite-state models, price is interpreted as the orthogonal projection of the terminal payoff vector onto information subspaces; in continuous-time models, the posterior density $q_t$ becomes a dynamic probability geometry, and price appears as a coordinate of projection within that geometry. Based on this structure, we derive the continuous financial field equation and the associated price propagation equation, showing that volatility should be endogenously linked to posterior uncertainty (i.e., posterior variance), rather than driven by exogenous “risk factors.” This provides a unified structural explanation for event-driven assets, the reconfiguration of implied distributions, and volatility patterns around major announcements. Finally, we derive a set of testable empirical implications. If market prices indeed reflect dynamic posterior geometry, several structural regularities should emerge: (i) price volatility should co-move with posterior uncertainty; (ii) posterior probabilities should adjust most rapidly in regions of maximal uncertainty and converge endogenously as certainty is approached; (iii) major informational constraints should induce systematic geometric reconfigurations in both prices and implied distributions; and (iv) part of what is commonly interpreted as risk premia can instead be understood as apparent drift arising from observing geometric motion under a coarse probability reference frame, rather than being entirely attributed to preference-based compensation.

The remainder of the paper is organized as follows. Section 2 develops the structural analogy between financial relativity and physical relativity. Section 3 presents the theoretical framework, including the core principles, discrete field equation prototypes, and the projection structure linking geometry and prices. Section 4 provides a discrete-state example illustrating the full chain from structural information to probability geometry and projected pricing. Section 5 develops the continuous-time financial field equation and shows how posterior geometry evolves under continuous information arrival, together with its empirical implications. Section 6 concludes.

\section{The Structural Analogy of Financial Relativity: From Physical Intuition to Pricing Geometry}\label{sec:analogy}

In finance, the Fundamental Theorem of Asset Pricing is usually understood as a mathematical statement about the absence of arbitrage: if the market is arbitrage-free, then there exists an equivalent measure $Q$ under which the discounted price process is a martingale. This statement is technically complete, yet conceptually it leaves a fundamental question unresolved: what exactly does this measure mean? If $Q$ is understood merely as a computational measure jointly induced by preferences and no-arbitrage conditions, then its form depends on particular preference specifications, making the interpretation of price relations path-dependent. Such an interpretation, however, does not adequately explain why price relations continue to display stable structural forms even when the probability measure, preference structure, or information partition changes. We therefore argue that the deeper structure of asset pricing is not exhausted by an “existence theorem”; rather, it contains a geometric mechanism that closely resonates with relativity: probability structure is not merely a tool for pricing formulas, but the geometry of the market itself; and price is not a simple average of future cash flows, but a projection of the terminal payoff vector onto the currently distinguishable information structure.

This claim is not merely metaphorical, but is supported by two facts. First, an asset price does not correspond directly to some “actual current location”; rather, it is the current representation of the terminal payoff $X$. Given a terminal state space $\Omega=\{\omega_1,\dots,w_n\}$, the ontological object of the asset is the terminal payoff vector $X=(X(\omega_1),\dots,X(w_n))$. At any time $t$, the market cannot fully distinguish all terminal states; it can only identify {\em certain equivalence classes of states} according to the current information partition $\mathcal F_t$. Accordingly, the current discounted price is not the terminal payoff itself, but its conditional expectation representation in the current information subspace. Second, under the probability geometry $Q$, the discounted price follows the simplest propagation law:
\begin{equation}
	\widetilde S_t=E^Q[X\mid \mathcal F_t].
\end{equation}
This means that, under the appropriate reference frame, no additional “risk-compensation force” is needed to drive prices; discounted prices simply propagate naturally along the directions permitted by the current information structure. At this point, a deep formal resonance between asset pricing and relativity already emerges: in physics, a free particle is not pushed by a force, but moves along a geodesic in spacetime geometry; in finance, the discounted price is likewise not dragged by some observable compensation force, but evolves through projection and propagation within probability geometry.

Accordingly, what we call “financial relativity” is not a mechanical transplantation of physical terminology into finance. It is a stronger claim: once terminal payoffs, information structure, probability measures, and price propagation are placed within a unified geometric framework, one finds that asset pricing theory already contains a relativity-like structural hierarchy. At a minimum, it includes four mutually corresponding ideas: first, there exists a unique inertial background that may be viewed as “flat spacetime”; second, market probability geometry is curved by terminal structural information; third, discounted prices exhibit local inertial propagation under the appropriate reference frame; and fourth, differences across reference frames cause the same price propagation to appear either as “motion without premia” or as “motion with premia,” just as different physical reference frames may represent the same worldline as either unaccelerated or accelerated. The following discussion develops these points in turn.

\subsection{Spacetime, Objects, and Projection: The Geometric Skeleton of Asset Pricing}\label{subsec:spacetime}

The basic objects of physical relativity are not isolated coordinate points, but events, worldlines, and metrics on a spacetime manifold. In financial markets, one must first dispel a possible misunderstanding: price itself is not the “object” of finance. Price is merely the market’s current compressed representation of future payoffs; it is not the object itself. More precisely, the “object” in finance is the terminal payoff vector $X$, which defines the value of an asset across all primitive terminal states. The price sequence $\{S_t\}$ is then the trajectory by which this terminal payoff is represented over time and under different levels of information resolution.

Under this definition, financial spacetime may be understood as a triadic structure consisting of time, the terminal state space, and information filtration. Formally, we write
\[
(t,\Omega,\mathcal F_t),\qquad \mathcal F_0\subset \mathcal F_1\subset \cdots \subset \mathcal F_T.
\]
Here, $\Omega$ specifies the possible terminal states, while $\mathcal F_t$ describes the blocks of states that the market can distinguish at time $t$. Information filtration is not an external description imposed on probability; it is part of financial spacetime itself: it determines what the market is able to “see” at a given moment. If the local structure of physical spacetime determines the natural trajectory of an object, then the information resolution of financial spacetime determines the form in which a terminal payoff can be represented at the present. From this point, the geometric meaning of price becomes clear. Let $L^2(\Omega,\mathcal F,Q)$ denote the Hilbert space weighted by $Q$, with inner product
\begin{equation}
	\langle Y,Z\rangle_Q :=E^Q[YZ].
\end{equation}
At each time, the representable subspace is
$L^2(\mathcal F_t)=\{Y\in L^2(\Omega,\mathcal F,Q): Y\text{ is }\mathcal F_t\text{-measurable}\}.
$
Then the current discounted price satisfies
\begin{equation}\label{eq:EQX}
	\widetilde S_t=E^Q[X\mid \mathcal F_t],
\end{equation}
and we show that the conditional expectation in (\ref{eq:EQX}) can be expressed as the orthogonal projection of the terminal payoff $X$ onto the subspace $L^2(\mathcal F_t)$. A rigorous proof is given in Proposition \ref{prop:projection-general} of the Appendix. This interpretation is crucial, because it allows the essence of price to be expressed in an extremely concise way: \emph{price is projection}. In this language, pricing is no longer described as “some weighted average,” but as the best geometric representation of the terminal payoff under the current information structure.

An additional advantage of the projection perspective is that it naturally characterizes the financial analogue of a “worldline.” The terminal payoff $X$ is fixed, but as information becomes progressively finer, the projection subspaces expand, so that
\[
\widetilde S_0=P_{\mathcal F_0}X,
\qquad
\widetilde S_1=P_{\mathcal F_1}X,
\qquad \dots,\qquad
\widetilde S_T=P_{\mathcal F_T}X=X.
\]
From this perspective, a price trajectory is not the movement of some object “itself,” but the changing sequence of projections of the same terminal payoff vector onto a nested family of subspaces. Moreover, the result of projection has the same dimension as the terminal payoff $X$ (see Proposition \ref{fig:projection} in the Appendix for further discussion). This means that the “motion” of the object in finance is not merely a change in coordinate values, but a change in representational structure itself. Figure \ref{fig:projection} in Appendix \ref{app:projection} provides an intuitive three-dimensional example illustrating how price evolution can be represented through changing projections.

\subsection{Inertial Frames, the Classification of Reference Frames, and the Financial Equivalence Principle}\label{subsec:frames}

Once the price process is understood as a sequence of projections, a natural question arises: under what probability structure does this projection-based propagation take its simplest form? Logically, this is the financial counterpart of the problem of inertial frames in physical relativity.

In physics, an inertial frame is not a coordinate choice given a priori; rather, it is defined through the form of motion: in an inertial frame, a free particle exhibits no additional acceleration, and its dynamics take the simplest possible form. By analogy, in finance one may define “natural propagation” as the situation in which the evolution of the discounted price contains no systematic drift. Formally, under a {\em chosen} probability measure $\pi$, this property can be written as
\begin{equation}
E^{\pi}[\widetilde R_t \mid \mathcal F_{t-1}] = 1,
\end{equation}
where $\widetilde R_t=\widetilde S_t/\widetilde S_{t-1}$ denotes the discounted return.

Accordingly, one may define a \emph{financial inertial frame} as a reference frame under whose associated probability geometry the discounted price process satisfies this simplest propagation law. As follows from (\ref{eq:EQX}), this property holds precisely under the equivalent martingale measure $\pi=Q$. Hence the probability geometry represented by $Q$ is not an arbitrary choice of measure, but the natural geometry under which price propagation takes inertial form. To proceed further, we first distinguish between two different levels of probability structure.

The first is the \emph{flat background} $P$: the prior geometry obtained by assigning equal probability weights to terminal states in the complete absence of structural information,
\begin{equation}
	P(w)=\frac1{|\Omega|},\qquad w\in\Omega.
\end{equation}

This geometry has the highest degree of symmetry and may be viewed as the unique “flat structure” of the financial world. It does not imply the existence of a true probability; rather, it represents the least committed description of the state space when no further constraints are imposed.

The second is the \emph{informational background} $Q$: the probability measure determined by the maximum entropy principle under terminal structural information constraints. More specifically, terminal states are often constrained by fundamentals, institutions, technologies, or external conditions. These constraints are not generated over time; instead, they impose structural restrictions on the terminal distribution. Under such constraints, the least biased probability assignment is the maximum entropy distribution satisfying them, thereby yielding $Q$. Hence
\[
Q = \arg\max H(Q)\quad \text{s.t. structural information constraints}.
\]

It is important to emphasize that the “information” referred to here is not the market information flow $\mathcal F_t$. The latter does not create new information, but merely reveals pre-existing terminal structure over time. In the sense of information conservation (to be discussed later), the system’s total information content is already fixed at the terminal date; what unfolds over time is only the revelation of that structure, not its generation. Therefore, $Q$ is a probability geometry determined once and for all by terminal constraints, rather than an object dynamically generated by information flow.

On this basis, we further distinguish the conditional geometry $Q_t$:
\[
Q_t(\cdot) = Q(\cdot \mid \mathcal F_t),
\]
which describes the conditional probability structure under the current filtration $\mathcal F_t$. Unlike $Q$, $Q_t$ varies over time, but this variation merely reflects the unfolding of information and does not alter the total amount of terminal information.

It follows that the “flatness” of finance depends on whether structural constraints are present: when no constraints exist, $Q=P$, and the geometry is flat; once constraints are present, $Q$ deviates from $P$, thereby creating a curved probability geometry. By contrast, “inertiality” depends on whether the observer describes price propagation under that geometry.

Based on these distinctions, one may classify financial reference frames into four types, whose structure is jointly determined by the informational background and the mode of observation.

The first is the \emph{inertial frame without gravity}. In this case, no structural constraints exist, so the market geometry is the flat background $Q=P$, and the observer also describes price propagation under this geometry. Therefore, discounted returns satisfy
\[
E[\widetilde R_t\mid \mathcal F_{t-1}]=E^P[\widetilde R_t\mid \mathcal F_{t-1}]=E^{Q}[\widetilde R_t\mid \mathcal F_{t-1}]=1,
\]
which corresponds to inertial motion in flat spacetime.

The second is the \emph{accelerated frame without gravity}. Here the geometry remains flat, that is, $Q=P$, but the observer uses an inconsistent probability geometry $Q'$, and therefore perceives an apparent drift. This corresponds to an accelerated reference frame in flat spacetime, in which acceleration arises from the choice of coordinates rather than from any real force.

The third is the \emph{static frame in a gravitational field}. In this case, terminal structural constraints imply that the market geometry is some $Q\neq P$, while the observer still uses the measure $P$ to describe price propagation. Then discounted returns appear as
\[
E[\widetilde R_t\mid \mathcal F_{t-1}] =E^{P}[\widetilde R_t\mid \mathcal F_{t-1}]= 1+\phi_t,
\]
where $\phi_t$ may be interpreted as the apparent drift generated by observing curved geometry from a non-natural reference frame.

The fourth is the \emph{freely falling frame in a gravitational field}. In this case, the market geometry is $Q$, and the observer also describes price propagation under $Q$, so that
\[
E^{Q}[\widetilde R_t\mid \mathcal F_{t-1}] = 1.
\]

This corresponds exactly to the freely falling frame in general relativity: the gravitational field has not disappeared, and the geometry remains curved, but the local first-order “gravitational acceleration” has been absorbed into the choice of reference frame, so that free motion again takes its simplest form. The financial version of the equivalence principle holds precisely in this sense: \emph{risk premia need not be treated as a fundamental force; they may instead be understood as apparent drift arising from observing curved probability geometry in a non-natural reference frame.}

One point deserves special emphasis. In physics, the equivalence principle does not imply that all gravitational effects can be eliminated; what can be locally removed is the first-order acceleration term, whereas tidal forces and curvature cannot. The same is true in finance: switching to $Q$ does not make uncertainty disappear. It only removes the systematic drift term from discounted prices; differences across state directions in volatility, covariance structure, and the multiplicity of $Q$ in incomplete markets remain as “second-order effects” of financial geometry. Financial relativity therefore does not deny risk or uncertainty. Rather, it claims that at least part of what is called a risk premium should be understood as a problem of geometry and reference frames, rather than being attributed a priori to the direct effect of preferences.

\begin{table}[htbp]
\centering
\caption{Structural correspondence between four types of physical and financial reference frames}\label{tab:frames}
\renewcommand{\arraystretch}{1.35}
\begin{tabular}{p{3.2cm}p{4.2cm}p{7.2cm}}
\toprule
Physical reference frame & Financial reference frame & Observed result and return representation \\
\midrule
Inertial frame without gravity & Market geometry is the flat background, $Q=P$ & The market adopts the flat geometry itself; discounted returns satisfy $E^{Q}[\widetilde R_t\mid\mathcal F_{t-1}]=1$, with no apparent premium. \\
Accelerated frame without gravity & Market geometry remains $Q=P$, but is evaluated under an incorrect market geometry & The geometry itself is not curved, but the inappropriate choice of reference frame produces an apparent deviation in discounted returns; the premium arises from mismeasurement rather than curvature. \\
Static frame in a gravitational field & Market geometry is $Q$, while the observer evaluates under $P$ & Discounted returns appear as $E^{P}[\widetilde R_t\mid\mathcal F_{t-1}]=1+\phi_t$; the risk premium is an apparent term generated by observing curved geometry under the wrong probability. \\
Freely falling frame in a gravitational field & Market geometry is $Q$, and the observer also evaluates under $Q$ & Discounted returns return to the form $E^{Q}[\widetilde R_t\mid\mathcal F_{t-1}]=1$; the first-order drift disappears, though uncertainty and curvature remain. \\
\bottomrule
\end{tabular}
\end{table}

\subsection{Relativistic Invariance, Metric, the Analogue of Light-Speed Invariance, and Curvature}\label{subsec:metric}

Relativity is not called relativity merely because it distinguishes reference frames, but because certain basic invariants are preserved under transformations between them. In physics, the laws of nature retain their form in inertial frames; general relativity extends this principle to a broader covariant setting. In finance, the true counterpart of such invariance is not the numerical value of price itself, but the geometric form of the pricing relation. Whether one starts from the conditional expectation relation of discounted prices or from the unit conditional expectation of discounted returns, the underlying claim is the same: given a probability geometry and an information structure, price is always the projection of the terminal payoff onto the currently representable subspace.

Accordingly, the relativistic invariance of finance can be expressed by the following principle: \emph{under any given probability geometry, the asset-pricing relation preserves its projection form.} More precisely, no matter which $Q$ is determined by informational constraints, the pricing equation is always
$
\widetilde S_t = E^Q[X\mid \mathcal F_t],
$
or equivalently,
$
E^Q[\widetilde R_t\mid \mathcal F_{t-1}] = 1,
$
. As in physical relativity, this invariance does not require numerical equality across reference frames; rather, it requires that different representations preserve the same structural relation.

Within this structure, the role of the metric is played by probability geometry. For any two terminal payoff vectors $Y,Z$, define the inner product
$
\langle Y,Z\rangle_Q = E^Q[YZ].
$
This inner product determines lengths, angles, and orthogonality in the price space. If all primitive states receive equal weight, the geometry is the flat background; if some states are systematically amplified under $Q$, then the geometric scale along those directions is correspondingly distorted. In this sense, “gravity” in finance should not be understood as an extra force acting on prices, but as a non-uniform reweighting of state directions. Put differently, \emph{the curvature of geometry does not first manifest itself as a change in price levels, but rather as a change in which state directions receive greater weight in pricing.}

In this same spirit, the “invariance of the speed of light” in physics also admits a rather precise counterpart in finance. In physics, the special role of the speed of light $c$ lies in the fact that it is invariant across all inertial frames and thereby determines the local structure of spacetime. In finance, the corresponding invariant is not a particular level of return, but the unit conditional expectation of discounted returns:
\[
E^Q[\widetilde R_t\mid \mathcal F_{t-1}] =\mathbf 1.
\]
Here the unit vector $\mathbf 1$ serves as the local benchmark of natural propagation. Its meaning is that once one switches to the appropriate geometric reference frame, the predictable component of discounted returns no longer depends on the asset class, but collapses uniformly to “1.” In this sense, $\mathbf 1$ is the local propagation constant of the financial world. It is not a price, a return, or an interest rate in the conventional sense, but a local invariant of pricing geometry.

Of course, what truly makes spacetime curved in relativity is not first-order acceleration, but curvature. Finance also has an analogous structure. Even when the first-order drift term disappears under $Q$, differences across state directions in volatility, covariance, incompleteness, and the multiplicity of probability geometries remain. These are “second-order effects” that cannot be eliminated by a simple change of reference frame. If the risk premium observed under $P$ corresponds to gravitational acceleration, then the non-uniformity of state weights and the path dependence induced by informational constraints are closer to the notion of curvature in finance. Financial relativity therefore does not merely replace old terms with new ones; it reorganizes the hierarchy within asset pricing theory: first-order apparent drift belongs to the problem of reference frames, whereas second-order non-uniformity and multiple geometries belong to the problem of genuine structure.

In summary, a systematic structural correspondence can be established between finance and relativity: terminal payoffs correspond to objects; time, state space, and information filtration constitute financial spacetime; conditional expectation projection corresponds to the local representation of price in the current geometry; the flat prior $P$ is the unique inertial background; probability geometry $Q$ corresponds to curved spacetime; the unit conditional expectation of discounted returns corresponds to a local propagation invariant; and changes of reference frame determine whether the same price propagation appears as “free motion” or as “accelerated motion with risk premia.” This structure is summarized in Table \ref{tab:relativity-map}.

\begin{table}[htbp]
\centering
\caption{Core conceptual correspondence between physical relativity and financial relativity}\label{tab:relativity-map}
\renewcommand{\arraystretch}{1.3}
\begin{tabular}{p{3.4cm}p{4.6cm}p{6.2cm}}
\toprule
Concept in physical relativity & Financial counterpart & Explanation \\
\midrule
Spacetime & Time, terminal state space, and information filtration $ (t,\Omega,\mathcal F_t)$ & Information filtration determines the state blocks distinguishable by the market at each moment. \\
Object & Terminal payoff vector $X$ & The ontological asset is not price, but the payoff structure over terminal states. \\
Worldline & Discounted price trajectory $\widetilde S_t$ & The continuous projection of the same terminal payoff onto information subspaces over time. \\
Inertial background & Equal-probability prior $P$ & The unique flat background of the financial world. \\
Curved spacetime & Probability geometry $Q$ & The systematic reweighting of state directions induced by terminal structural information. \\
Metric & Inner product $\langle Y,Z\rangle_Q=E^Q[YZ]$ & Determines lengths, orthogonality, and projection in price space. \\
Free motion / geodesic & Martingale propagation of discounted prices under $Q$ & Natural propagation without systematic drift. \\
Equivalence principle & Elimination of first-order apparent premia under $Q$ & Risk premia can be interpreted as apparent terms arising in non-natural reference frames. \\
Light-speed invariance & $E^Q[\widetilde R_t\mid\mathcal F_{t-1}]=1$ & The unit vector $\mathbf 1$ defines local natural propagation. \\
Gravitational acceleration & Risk premium $\phi_t$ observed under $P$ & First-order apparent drift generated by observing curved geometry under the wrong measure. \\
Curvature & Non-uniform state weights, volatility–covariance structure, and multiple $Q$ & Second-order structure that cannot be eliminated by a simple change of reference frame. \\
\bottomrule
\end{tabular}
\end{table}

\section{Financial Relativity: Theoretical Principles, Field Equations, and Philosophical Implications}

The discussion in Section 2 shows that the relationship between asset pricing theory and physical relativity is not merely metaphorical, but rests on a deeper structural isomorphism: in both theories, \emph{geometry} is not an auxiliary mathematical device, but the fundamental carrier of the law of motion. By replacing the traditional “force–background” dichotomy with geometric structure, physical relativity reinterprets gravity from an external force into the curvature of spacetime itself. As argued in the previous section, asset pricing theory can likewise be reformulated in geometric terms: probability measures are no longer merely tools for computing expectations, but constitute the probability geometry on which price propagation depends; and the price process is no longer interpreted as the outcome of exogenous driving forces, but as a natural mode of motion under that geometry. In this way, the relation among probability measures, price projection, and dynamic propagation can be unified within a single geometric framework.

The central proposition of this paper can be stated as follows: \emph{terminal structural information determines how probability geometry is curved, while probability geometry determines how information is manifested in prices.} This proposition can be understood at three interrelated levels. First, the market does not evolve in a pre-given and uniform probability space; rather, the structural information embedded in terminal states changes the relative weights of different state directions through constraints, thereby potentially generating a non-uniform probability geometry. Second, once probability geometry is given, price evolution no longer depends on exogenous “risk-premium forces”; its natural form is instead expressed as conditional expectation projection and local inertial propagation under that geometry. Third, prices are not merely passive reflections of terminal information, but the main mechanism through which information is manifested in time: through price paths, the market gradually unfolds the structural differences embedded in terminal states into the observable domain. Hence, within this framework, terminal structural information, probability geometry, and price dynamics are no longer three separate objects, but three aspects of one unified theoretical structure: information determines geometry, geometry governs prices, and the evolution of prices realizes the manifestation of information in the market.

We regard the foregoing formulation as important because it changes the causal order of asset pricing theory. In the classical formulation, probability, preferences, and prices tend to be placed on the same plane: given a physical probability $P$ and preferences or cognitive biases, prices arise from some discounted average, while the risk-neutral measure $Q$ is merely a convenient rewriting of that process. Financial relativity, by contrast, argues that at a more fundamental level one must first answer the question: under what geometry does the market understand the terminal state space? Once this is clarified, ideas that were previously scattered across no-arbitrage pricing, stochastic discount factors, risk-neutral measures, filtering theory, and information economics can be brought into a single theoretical language.

\subsection{Basic Principles: From Terminal Structure to Probability Geometry to Price Motion}

To avoid conceptual confusion, we begin by distinguishing between two objects that are both commonly referred to as “information.” The first is \textbf{terminal structural information}, which characterizes structural constraints in the terminal state space, assigns unequal weights to different state directions, and thereby determines the curvature of probability geometry. The second is the \textbf{revelation process}, which describes how these structural differences are gradually manifested through price paths over time. The former constitutes the source of geometry, while the latter represents the unfolding of geometry in time.

In a closed system, the total amount of terminal information remains unchanged; the dynamics of prices and probabilities merely redistribute this pre-existing information rather than create new information. Formally, the amount of information manifested through prices at each period plus the remaining uncertainty is equal to the initial total uncertainty. This law of information conservation will be established in Appendix \ref{app:info-conservation} in the form of an entropy decomposition (see Theorem \ref{thm:info-conservation}). Accordingly, the terms “revelation” or “manifestation” used in this paper always refer to the unfolding of pre-existing structural information in the market, rather than to its generation.

Under the above distinction, the first principle of financial relativity can be stated as follows: \textbf{terminal structural information determines how probability geometry is curved.} Let the terminal state space be
$
\Omega=\{\omega_1,\omega_2,\dots,w_n\},
$
and let the symmetric prior in the absence of structural and informational constraints be
$
P(w_i)={1}/{n},\qquad i=1,\dots,n.
$
$P$ defines a uniform probability measure and thus corresponds to the unique flat probability structure. Since it does not incorporate informational constraints, it will be interpreted as a {\em “coarse-grained probability”}. Once terminal constraints are introduced, the probability measure deviates from $P$, thereby generating a curved geometry. In financial relativity, the market probability measure $Q$ is interpreted as a probability geometry determined by terminal structural information through the maximum entropy principle, rather than being explained by preferences or cognitive biases. The inner product defined by $Q$,
$
\langle X,Y\rangle_Q := E^Q[XY],
$
endows the terminal payoff space with a Hilbert structure, so that lengths, angles, and projections are all characterized by probability geometry.

The second principle states that \textbf{probability geometry determines how prices move.} Under this geometry, the discounted price is defined by
$
\widetilde S_t = E^Q[X\mid \mathcal F_t],
$
which is equivalent to the orthogonal projection of the terminal payoff vector onto the information subspace:
\begin{equation}
	\widetilde S_t=P_{\mathcal F_t}X.
\end{equation}
Hence, price is not the result of exogenous driving forces, but the projection representation under geometric structure; its dynamics are expressed as the continuous updating of projections on increasingly refined information subspaces. On this basis, the natural propagation of prices can be written as
$
E^Q[\widetilde R_t\mid \mathcal F_{t-1}] =1,
$
where $\widetilde R_t$ denotes the discounted return; in continuous time, it becomes
\begin{equation}
	\mathrm{d}\widetilde S_t = \Sigma_t\,\mathrm{d}W_t^Q,
\end{equation}
where $W_t^Q$ is the innovation process under geometry $Q$. This equation has the same form as the “risk-neutral dynamics” of classical finance, but its theoretical interpretation is entirely different: in the classical framework, it is merely a measure-theoretic device introduced for pricing; in financial relativity, it is the inertial propagation equation of the asset in its natural geometry. Accordingly, the risk premium $\phi_t$ is no longer understood as an intrinsic property of the asset itself, but as an apparent acceleration that arises when inertial motion is observed from a non-natural reference frame. If one uses a coarse-grained probability $P$ to observe the same asset, one obtains
$
\mu_t = r_f + \phi_t,
$
so that returns appear to be driven by an additional force and to contain a premium; but from the viewpoint of geometry $Q$, this apparent force disappears, leaving only unpredictable local noise. This corresponds to the financial version of the equivalence principle: \emph{the “risk premium” embedded in drift can be eliminated by choosing the correct geometric reference frame.}

\subsection{Financial Field Equations: The Origin of Geometry and the Selection of Multiple Geometries}

In physical relativity, geometry is linked to its source by field equations. Correspondingly, financial relativity must not only describe price motion under a given probability geometry $Q$, but also explain how this geometry is determined by terminal structure. This leads to the prototype of the financial field equation:
\begin{equation}
	\mathcal{G}(Q)=\mathcal{I},
\end{equation}
where $\mathcal{G}(Q)$ denotes the geometric operator generated by the market probability geometry $Q$, and $\mathcal{I}$ denotes the source term constituted by terminal structural information. In structural terms, this expression corresponds to the Einstein field equation
\begin{equation}
	G_{\mu\nu}=8\pi T_{\mu\nu}, \qquad \text{set }c=1
\end{equation}
Its central meaning is that probability geometry is not given a priori, but is endogenously generated by terminal structural constraints. The basic logic of financial relativity may therefore be stated as follows: terminal informational structure shapes probability geometry, and probability geometry organizes price motion (that is, the manifestation of terminal information).

To obtain an operational representation, one may introduce a logarithmic geometric potential function $\phi$, and write the probability measure as
\begin{equation}
	Q(w)=\frac{P(w)e^{\phi(w)}}{\sum_{u\in\Omega}P(u)e^{\phi(u)}},
\end{equation}
where $P$ is the flat prior. The function $\phi$ describes the deviation of probability geometry from the flat structure. In a discrete state space, a natural prototype equation is
\begin{equation}
	L\phi = \kappa\rho,
\end{equation}
where $L$ is the discrete Laplacian, $\rho$ denotes the source term of terminal structural information, and $\kappa$ is the coupling strength. This equation shows that the role of structural sources is not to enter the pricing equation directly as an “external force,” but to distort probability geometry by changing the geometric potential $\phi$, thereby indirectly determining the mode of price propagation.

The significance of this representation lies in elevating the problem of generating probability geometry to an extensible operator framework. First, the geometric operator $\mathcal{G}$ need not be restricted to linear or static forms; in continuous time it may be generalized to a stochastic partial differential operator on log densities, or to a nonlinear operator coupled with the evolution of posterior distributions. Second, the source term $\mathcal{I}$ should not be understood narrowly as an exogenous signal, but more generally as any structural constraint capable of breaking equal-probability symmetry, including institutional constraints, accounting constraints, technological paths, credit constraints, and behavioral structure. Financial field equations are therefore not a closed model, but a theoretical prototype capable of accommodating different structural sources.

Unlike physical spacetime, probability geometry is generally not unique in incomplete markets. The no-arbitrage condition guarantees only that the set of equivalent martingale measures
\[
\mathcal{Q}=\{Q:Q\text{s.t. No. Arbitrage}\}
\]
is non-empty, but it is insufficient to determine a unique element within it. Thus, under financial relativity, the flat reference frame $P$ is unique, while the curved geometry $Q$ may constitute a set. This difference implies that the problem of generating geometry includes not only “how geometry is determined by structure,” but also “how one selects from the feasible set.”

On this basis, the maximum entropy principle can be introduced as a normalization device for geometric selection: among all probability geometries satisfying the constraints, choose
\[
Q^*=\arg\max_{Q\in\mathcal{Q}} H(Q),
\]
where $H(Q)$ denotes Shannon entropy. This principle expresses a criterion of “minimal structural bias,” namely that under existing structural constraints, no additional asymmetry should be introduced. Its function is to provide a closure mechanism for multiple geometries without relying on unobservable preference specifications.

Methodologically, this selection principle has a clear meaning: when geometry is not unique, the problem is no longer one of characterizing some latent structure of risk aversion, but of identifying the most neutral probability geometry under given constraints. In this way, the source of the probability measure $Q$ shifts from a preference-based explanation to one based on structure and information constraints. Financial relativity thus removes the problem of generating geometry from the realm of unobservable utility systems and relocates it within a structural framework that can be modeled and tested.

\subsection{Price as Geometric Projection: The Observational Structure of Information Manifestation}

Under the two basic principles and the field equation $\mathcal G(Q)=\mathcal I$, price is no longer an independently specified dynamic variable, but becomes the directly observable image of probability geometry under the information filtration $\mathcal F_t$. Accordingly, the central issue of this subsection is not to reinterpret the pricing formula itself, but to characterize \emph{how the price mapping determined by probability geometry organizes the observable manifestation of terminal structural information in the market.} Given a geometry $Q$, the discounted price
$
\widetilde S_t = E^Q[X\mid \mathcal F_t]
= P_{\mathcal F_t}X.
$
compresses the terminal payoff function $X$ from the full space into the current information subspace by means of projection. Probability geometry $Q$ determines the projection structure, and the projection structure determines which terminal differences can be distinguished by the current price.

\begin{proposition}[Price as the mechanism through which terminal structural information is manifested]\label{prop:price-reveals-information}
Let the terminal state space be $\Omega$, the probability geometry be $Q$, and the terminal payoff be a function $X:\Omega\to\mathbb R$. Define the discounted price at time $t$ by
\[
\widetilde S_t(w)=E^Q[X\mid \mathcal F_t](w).
\]
For any observable price value $s$, define the corresponding set of terminal states by
\[
\Omega_t(s):=\{\,w\in\Omega:\widetilde S_t(w)=s\,\}.
\]
Then observing the price $\widetilde S_t=s$ is equivalent to knowing that the true state $w$ belongs to the set $\Omega_t(s)$, and the corresponding conditional probability is
\[
Q(w\mid \widetilde S_t=s)
=
\frac{Q(w)}{\sum_{u\in\Omega_t(s)}Q(u)},
\qquad w\in\Omega_t(s),
\]
while for $w\notin\Omega_t(s)$, one has $Q(w\mid \widetilde S_t=s)=0$. The corresponding residual uncertainty is measured by
\[
H_t(s)
:=
-\sum_{w\in\Omega_t(s)}
Q(w\mid \widetilde S_t=s)\log Q(w\mid \widetilde S_t=s).
\]
\end{proposition}

\begin{proof}
Since $\widetilde S_t$ is an $\mathcal F_t$-measurable function, observing the price value $s$ is equivalent to knowing that the true state $w$ falls into the set
\[
\Omega_t(s)=\{w\in\Omega:\widetilde S_t(w)=s\}.
\]
On this event, the conditional probability follows directly from Bayes’ formula:
\[
Q(w\mid \widetilde S_t=s)
=
\frac{Q(w)}{Q(\Omega_t(s))}
=
\frac{Q(w)}{\sum_{u\in\Omega_t(s)}Q(u)},
\qquad w\in\Omega_t(s),
\]
while for states outside this set the conditional probability is zero. By the definition of Shannon entropy, the residual uncertainty conditional on the price is therefore
\[
H_t(s)
=
-\sum_{w\in\Omega_t(s)}
Q(w\mid \widetilde S_t=s)\log Q(w\mid \widetilde S_t=s).
\]
This proves the proposition.
\end{proof}

Proposition \ref{prop:price-reveals-information} shows that price does not directly reveal the terminal state itself. Rather, through the projection operator, it induces an equivalence partition $\{\Omega_t(s)\}$ on the state space. This partition characterizes the finest structure that the market can distinguish under the current geometry and informational constraints. In other words, what price manifests is not “the terminal state itself,” but “which price-equivalence class the terminal state belongs to.”

This structure immediately implies a basic inequality. Since the price $\widetilde S_t$ is an $\mathcal F_t$-measurable function, it constitutes only a functional compression of the information filtration. Hence, in information-theoretic terms,
\[
I(\widetilde S_t)\le I(\mathcal F_t),
\]
and the inequality is generally strict. Equivalently, on the state space, the information filtration $\mathcal F_t$ corresponds to a finer partition, while price corresponds only to a coarser partition. Therefore, in expectation one has
\[
E^Q[H_t(s)] \ge H_{N|t},
\]
where $H_{N|t}$ denotes the conditional entropy under the full information filtration (i.e., when one knows which specific branch the system lies in at time $t$; see the conditional entropy expression in Theorem \ref{thm:info-conservation}). It should be noted, however, that for a particular value $s$, $H_t(s)$ need not be pointwise greater than $H_{N|t}$.

It follows that the price path is not only the evolution of a projected vector, but can also be understood as the evolution of a sequence of state partitions organized by probability geometry. As time passes, the refinement of the information filtration $\mathcal F_t$ changes the projection operator, which in turn induces an increasingly fine partition of the state space, so that terminal structural information gradually enters the observable domain of the market.

Connecting this mechanism to the field equation yields the complete observational picture of financial relativity, shown in Figure \ref{fig:tujing}:
\begin{figure}[H]
\centering
  \includegraphics[width=0.9\textwidth]{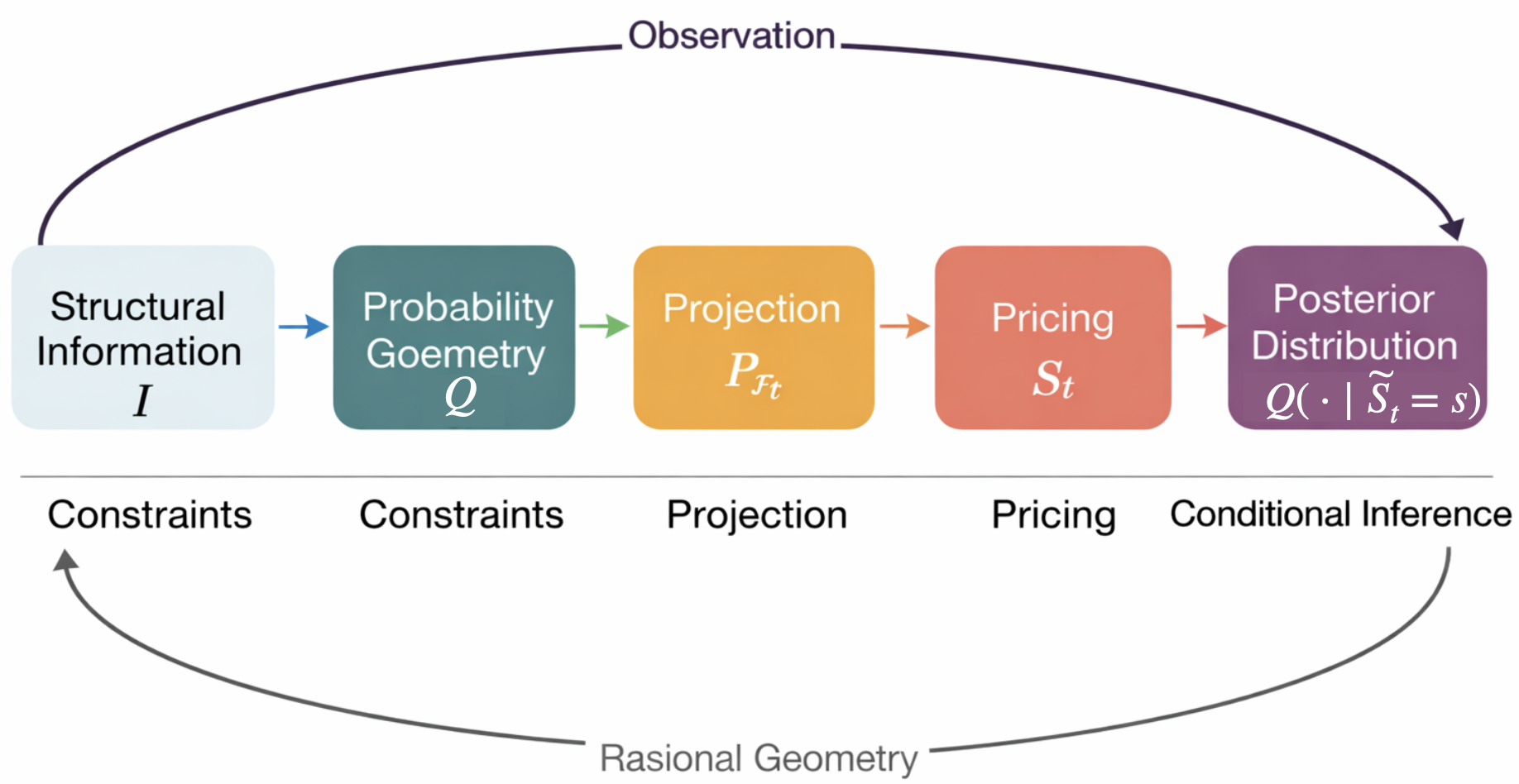}
  \caption{}\label{fig:tujing}
\end{figure}

That is, terminal structural information first generates probability geometry through constraints; probability geometry determines projection structure; projection structure further determines the price path; and price, as an observable quantity, in turn induces a conditional distribution over terminal states. Price is therefore not information itself, but the observable representation of information under geometric constraints.

In this sense, the statement that “probability geometry determines how information is manifested in prices” may be formulated as a proposition with clear structural meaning: probability geometry determines not only the form of price motion, but also the degree to which terminal structural information becomes distinguishable when it enters the market. More specifically, given a geometry $Q$, the equivalence partition induced by the price mapping compresses the state space, so that some terminal differences are grouped into the same price value and thus remain indistinguishable, while other differences are separated in advance and enter the observable domain. The “manifestation” of information is therefore not uniform, but unfolds gradually along the discriminating structure prescribed by probability geometry.

This perspective reveals the fundamental difference between financial relativity and traditional asset pricing theory. In the classical framework, price is usually understood as a function of probability and preferences, and the probability measure $Q$ is treated as an auxiliary computational object. In the present framework, by contrast, price is understood as a projection operator in probability geometry, and its dynamics reflect the gradual manifestation of terminal structural information under geometric constraints. The probability measure $Q$ is therefore no longer merely a pricing device, but becomes the generative mechanism of the market’s observable structure; and price is no longer merely a numerical sequence, but the geometric expression of terminal structure in time.

\section{Numerical Illustration: From Abstract Principles to a Computable Model}

The previous section formulated the core proposition of financial relativity in conceptual terms: terminal structural information determines the probability geometry of the market; probability geometry determines the natural motion of the price vector; and the evolution of the price vector determines how terminal structure is progressively revealed in the market. To transform these abstract statements into a visible, computable, and interpretable structure, this section constructs a concrete example that illustrates how terminal structure first distorts probability geometry and then alters pricing through geometric projection. For pedagogical clarity, we adopt a finite-state branching structure to represent the chain of “structural source—geometry—projection pricing—information revelation.”

We consider a three-period finite-state market. The terminal time is $T=2$, and the terminal state space is
\[
\Omega=\{\omega_1,\omega_2,\dots,w_8\}.
\]
At the intermediate time $t=1$, the market can only distinguish two branches:
\[
A=\{\omega_1,\omega_2,\omega_3\}, \qquad B=\{w_4,w_5,w_6,w_7,w_8\}.
\]
Accordingly, the information filtration is given by
\[
\mathcal F_0=\{\emptyset,\Omega\},\qquad \mathcal F_1=\sigma(A,B),\qquad \mathcal F_2=2^\Omega.
\]
This structure provides the simplest multi-period implementation of the “expanding information subspace” discussed earlier: at $t=0$, all terminal state directions are compressed into a constant vector; at $t=1$, the market can distinguish between two branches but not among states within each branch; at $t=2$, all structural differences are fully revealed.

\begin{itemize}

\item {\bf Step 1: Flat inertial frame.} Under the theoretical stance of this paper, if terminal primitive events are not distinguished by any structural information, the corresponding probability geometry is the unique uniform prior:
\[
P(w_i)=\frac18,\qquad i=1,\dots,8.
\]
Under this flat geometry,
\[
P(A)=\frac38,\qquad P(B)=\frac58.
\]
It is important to emphasize that these proportions do not indicate that the market “favors” or “disfavors” a branch; they merely reflect the mechanical aggregation of equally weighted primitive events. In other words, $P$ is the flat background of the state space rather than the market’s true understanding of terminal structure.

\item {\bf Step 2: Terminal structural source.} The introduction of structural information serves to curve the geometry. For transparency, we consider a block-symmetric source: terminal states within each branch remain undistinguished, while structural asymmetry is introduced across branches. Specifically, define
\[
\rho(w)=
\begin{cases}
\rho_A,& w\in A,\\
\rho_B,& w\in B,
\end{cases}
\qquad \text{with}\qquad 3\rho_A+5\rho_B=0.
\]
This balance condition ensures compatibility with the zero-mean normalization of the geometric potential; see Appendix \ref{app:block-symmetric-field} for a rigorous derivation. It does not imply that positive and negative information cancel out, but rather that the geometric potential rearranges probability mass without altering the total mass. For concreteness, we set
\[
\rho_A=5,\qquad \rho_B=-3.
\]
This means that terminal structure assigns relatively higher weight to branch $A$ and relatively lower weight to branch $B$. This “support” should not be interpreted in a subjective sense, but as reflecting that states in $A$ receive higher structural weight under logical, institutional, or fundamental constraints.

\item {\bf Step 3: Discrete field equation and geometric potential.} Following the field equation prototype, geometry is generated indirectly through a geometric potential $\phi$. Let
\[
Q(w)=\frac{P(w)e^{\phi(w)}}{\sum_{u\in\Omega}P(u)e^{\phi(u)}},
\]
where $\phi$ captures the deviation from the flat prior. In this example, we adopt a discrete Poisson-type equation:
\[
L\phi=\kappa\rho,
\]
where $L$ is the graph Laplacian and $\kappa>0$ is a coupling parameter. Appendix \ref{app:block-symmetric-field} shows that under block symmetry,
\[
\phi(w)=
\begin{cases}
\phi_A,& w\in A,\\
\phi_B,& w\in B,
\end{cases}
\qquad
\phi_A=\frac58\kappa,\qquad \phi_B=-\frac38\kappa.
\]
Thus, the geometric potential reorganizes the originally uniform state space into two layers: higher weight within $A$ and lower weight within $B$.

\item {\bf Step 4: Formation of probability geometry.} Substituting the potential into the definition yields
\[
Q(w_i)=\frac{e^{5\kappa/8}}{3e^{5\kappa/8}+5e^{-3\kappa/8}},\qquad i\in A,
\]
and
\[
Q(w_i)=\frac{e^{-3\kappa/8}}{3e^{5\kappa/8}+5e^{-3\kappa/8}},\qquad i\in B.
\]
Setting $\kappa=0.4$, we obtain
\[
Q(A)\approx 0.4722,\qquad Q(B)\approx 0.5280,
\]
compared to
\[
P(A)=0.375,\qquad P(B)=0.625.
\]
This illustrates the central geometric idea: probability is not arbitrarily adjusted, but systematically distorted by structural information through the geometric potential.

\item {\bf Step 5: Price as projection.} Consider an asset with terminal payoff
\[
X=(12,10,8,6,4,3,2,1).
\]
Ignoring discounting, the initial price is
\[
S_0=E^Q[X]\approx 6.412,
\]
while under $P$,
\[
E^{P}[X]=5.75.
\]
The difference arises purely from geometric distortion.

At $t=1$,
\[
S_1 = E^Q[X\mid \mathcal F_1]
    = (10,10,10,3.2,3.2,3.2,3.2,3.2).
\]
Thus,
\[
10=E^Q[X\mid A],\qquad 3.2=E^Q[X\mid B].
\]

\item {\bf Step 6: Information revelation through prices.} Since $S_1\in\{10,3.2\}$,
\[
\Omega_1(10)=A,\qquad \Omega_1(3.2)=B.
\]
Hence,
\[
Q(w\mid S_1=10)=Q(w\mid A),\qquad
Q(w\mid S_1=3.2)=Q(w\mid B).
\]
With symmetry,
\[
Q(w\mid A)=\frac13,\quad Q(w\mid B)=\frac15.
\]

Thus,
\[
H_t(10)=\log 3,\qquad H_t(3.2)=\log 5.
\]

Expected conditional entropy:
\[
E[H_t(S_t)] = Q(A)\log 3 + Q(B)\log 5.
\]

Total entropy under $Q$:
\[
H_2 \approx 2.9716\ \text{bits},
\quad
E[H_t(S_t)]\approx 1.9738\ \text{bits}.
\]

Thus revealed information:
\[
I(S_1)\approx 0.9978\ \text{bits}.
\]

This shows that price reveals about one bit of terminal structure.

\end{itemize}


In the conventional interpretation, prices exceed simple averages because the market requires compensation for risk; however, this explanation does not clarify why such compensation should manifest itself in precisely this particular structure over the state space. Financial relativity offers a different answer: prices change because the geometry changes, and the geometry changes because terminal structural sources alter the relative weights of state directions. The advantage of this perspective is that it transforms what would otherwise be an unobservable inference about preferences into a structural source that can be discussed, modeled, and extended.

At the same time, prices gradually release the information embedded in this geometry through projection operators, and this process can be measured precisely using entropy and mutual information. The asset pricing problem is therefore reframed—from one of “how to value assets given preferences” to one of “how information enters prices through geometry under given structural constraints.”

\section{Continuous Financial Field Equation: From Posterior Geometry to Price Geodesics}

Section 3 has established two fundamental principles of financial relativity. First, terminal structural information determines how the market probability geometry is curved; second, probability geometry determines how the price vector evolves. In the discrete setting, these principles are expressed through
\[
L\phi=\kappa\rho,
\qquad
Q(w)\propto P(w)e^{\phi(w)}
\]
and
\[
\widetilde S_t=E^Q[X\mid\mathcal F_t].
\]

However, the above discrete framework implicitly relies on a key assumption: terminal structural information is specified once and for all as a set of constraints, and the system does not receive new structural constraints during its evolution. As a result, total uncertainty remains constant. This corresponds to a closed system.

We now turn to a different setting: the market continuously receives structural constraints from outside during its evolution, so that terminal uncertainty itself is progressively recharacterized over time. In this case, probability geometry is no longer a once-for-all object, but becomes a dynamically evolving entity. In this sense, the continuous-time model does not describe the gradual revelation of pre-existing information, but rather the \emph{progressive introduction of structural constraints}. Accordingly, the discrete statement that “the structural source determines geometry” must be reformulated in continuous time as follows:
\emph{the dynamic updating of structural constraints reshapes posterior geometry, which in turn determines price propagation.}

The main objective of this section is to formalize this statement rigorously and to show why it can be interpreted as the field equation in continuous financial relativity.

\subsection{From Prior Geometry to Posterior Geometry}

Let the terminal payoff be represented by a real-valued random variable $X$ with prior density
$
p(x), \; x\in\mathbb R.
$
Here $p$ represents the flat probability geometry in the absence of structural constraints. As in the discrete case, it is not interpreted as an empirical “true” probability, but rather as the most symmetric reference frame.

In the continuous setting, the market does not observe $X$ directly. Instead, it receives an observation process driven by external structural constraints:
\[
d\xi_t=\sigma X\,dt+dB_t,
\qquad 0\le t\le T,
\]
where $B_t$ is a standard Brownian motion and $\sigma>0$ measures the strength of the constraints. Let
$
\mathcal F_t=\sigma(\xi_s:0\le s\le t),
$
so that $\mathcal F_t$ represents the accumulation of structural constraints over time.

Under this structure, the market probability geometry is defined as the posterior density
\[
q_t(x):=P(X=x\mid\mathcal F_t),
\]
that is, the least biased estimate of the state space under current constraints.

By Bayes’ rule,
\[
q_t(x)=
\dfrac{
p(x)\exp\!\left(\sigma x\,\xi_t-\frac12\sigma^2x^2t\right)
}{
\int_{\mathbb R}
p(y)\exp\!\left(\sigma y\,\xi_t-\frac12\sigma^2y^2t\right)\,dy
}.
\]

This expression shows that the posterior geometry remains an exponential tilt of the prior geometry:
\begin{equation}\label{eqn:iienagg}
	q_t(x)=\frac{p(x)e^{\Phi_t(x)}}{Z_t},
\qquad
\text{where}\,\,\,
\Phi_t(x):=\sigma x\,\xi_t-\frac12\sigma^2x^2t.
\end{equation}

The crucial difference from the discrete case is that the geometric potential $\Phi_t$ is no longer a static object determined by a one-shot structural source, but a stochastic field that evolves continuously with the arrival of constraints. As a result, probability geometry itself becomes a dynamic object.

\subsection{A Rigorous Formulation of the Continuous Financial Field Equation}

Define the posterior mean
\[
m_t:=E^Q[X\mid\mathcal F_t]=\int_{\mathbb R}x\,q_t(x)\,dx.
\]

Define the innovation process
\[
W_t^Q:=\xi_t-\sigma\int_0^t m_s\,ds,
\qquad
dW_t^Q=d\xi_t-\sigma m_t\,dt.
\]

By the innovation theorem in filtering theory, $W_t^Q$ is a Brownian motion with respect to $\mathcal F_t$. This decomposition implies that the observation process can be written as
\[
d\xi_t=\sigma m_t\,dt+dW_t^Q,
\]
where the first term is the predictable component implied by the current geometry, and the second term represents the unpredictable arrival of new constraints.

\begin{theorem}[Continuous Financial Field Equation]\label{thm:continuous-field-equation}
Under the above structure, the posterior density satisfies
\[
dq_t(x)=\sigma(x-m_t)q_t(x)\,dW_t^Q.
\]
\end{theorem}

This equation describes how probability geometry is continuously reshaped as structural constraints are updated. The proof is given in Appendix \ref{app:theoremkkk}.

By comparison, the discrete field equation
\[
L\phi=\kappa\rho,
\qquad
Q(w)\propto P(w)e^{\phi(w)}
\]
represents the statement that “the structural source determines geometry,” where $\rho$ is the terminal structural source, $\phi$ the geometric potential, and $Q$ the resulting probability geometry. In that setting, constraints are specified once at the initial time, and the system evolves within a fixed geometry thereafter.

In the continuous case, this structure becomes
\[
dq_t(x)=\sigma(x-m_t)q_t(x)\,dW_t^Q,
\qquad
q_t(x)\propto p(x)e^{\Phi_t(x)}.
\]

Here, the mechanism of geometry formation changes fundamentally: structural constraints are no longer given once and for all, but enter the system gradually through the innovation term $dW_t^Q$, continuously reshaping the posterior geometry $q_t$. Meanwhile, $m_t$ provides the local center of the current geometry, ensuring that updates are anchored to the existing structure.

From a theoretical perspective, Theorem \ref{thm:continuous-field-equation} does not describe how prices evolve given a geometry, but rather addresses a more fundamental question—\emph{how geometry itself is generated and updated under the continuous introduction of constraints}. This is precisely analogous to the role of the Einstein field equation in general relativity: it does not describe geodesic motion, but determines how spacetime geometry is generated by matter. Likewise, the equation above is not a price evolution equation, but a law governing how market probability geometry is shaped by structural constraints.

\subsection{From Field Equation to Price Propagation}

Once the dynamics of geometry are specified, price evolution is no longer an independent object, but a consequence of the field equation. Define
\[
S_t=e^{-r_f(T-t)}E^Q[X\mid\mathcal F_t]
=
e^{-r_f(T-t)}m_t.
\]

Differentiating,
\[
\mathrm{d}m_t=\int_{\mathbb R}x\,\mathrm{d}q_t(x)\,dx
=
\sigma\int_{\mathbb R}x(x-m_t)q_t(x)\,\mathrm{d}x\,dW_t^Q.
\]

Since
\[
\int_{\mathbb R}x(x-m_t)q_t(x)\,dx
=
E^Q[X^2\mid\mathcal F_t]-m_t^2
=
\operatorname{Var}^Q(X\mid\mathcal F_t),
\]
we obtain
\[
dm_t=\sigma\,\operatorname{Var}^Q(X\mid\mathcal F_t)\,\mathrm{d}W_t^Q.
\]

Hence the price satisfies
\begin{equation}
\mathrm{d}S_t=r_fS_tdt+\Sigma_t\,\mathrm{d}W_t^Q,
\qquad
\Sigma_t=e^{-r_f(T-t)}\sigma\,\operatorname{Var}^Q(X\mid\mathcal F_t).
\end{equation}

This result has a clear structural interpretation: the price geodesic equation is not a primitive assumption, but a consequence of the field equation; volatility is not exogenously specified, but is endogenously generated by the uncertainty of the current probability geometry. This provides a precise mathematical formulation of the statement that “geometry determines how prices move.”

\subsection{A Continuous Model: A Simple Implementation of the Financial Field Equation}

To illustrate the above structure, consider a simple setting with clear economic meaning: the terminal fundamental takes only two possible values, while the market continuously receives new structural constraints and updates its probability assessment accordingly.

Such a structure is widespread in real financial markets. Many economic and financial outcomes are discrete in nature, yet the associated probabilities are not known in advance; rather, they are continuously revised over time. Market participants do not face a static known distribution, but an evolving constraint environment shaped by institutional arrangements, financing conditions, policy expectations, disclosure schedules, and trading behavior. In this process, the terminal states themselves remain unchanged, but structural information about them continuously enters the system, leading to a reallocation of probability weights. Probability is therefore no longer an exogenous object, but a posterior geometry that evolves under the ongoing introduction of constraints; prices and volatilities are the observable manifestations of this geometric adjustment.

For simplicity, let
$
X\in\{L,H\},\quad H>L,
$
and suppose the market observes
$
\xi_t=\sigma Xt+B_t.
$
Let the posterior probability be
$
\pi_t=Q(X=H\mid\mathcal F_t).
$

From Theorem \ref{thm:continuous-field-equation},
\[
d\pi_t
=
\sigma(H-L)\pi_t(1-\pi_t)dW_t^Q.
\]

This equation shows that the speed of probability updating is governed by current uncertainty. When $\pi_t\approx 1/2$, uncertainty is maximal and the geometry is most sensitive to new constraints; when $\pi_t$ approaches $0$ or $1$, the geometry stabilizes and becomes less responsive. The discounted price is
$
\widetilde S_t=L+(H-L)\pi_t.
$

Thus,
\[
d\widetilde S_t
=
\sigma(H-L)^2\pi_t(1-\pi_t)dW_t^Q.
\]

The corresponding volatility is
$
\Sigma_t=\sigma(H-L)^2\pi_t(1-\pi_t).
$
This yields the key structural relation:
\[
\Sigma_t\propto \pi_t(1-\pi_t),
\]
namely, price volatility is endogenously determined by current uncertainty. Importantly, within this open-system framework, the evolution of $\pi_t$ should not be interpreted as the gradual revelation of pre-existing information, but as the response of posterior geometry to the continuous introduction of structural constraints. Accordingly, price dynamics reflect not only changes in uncertainty, but also the manner and timing with which constraints enter the system.

This prediction is consistent with the volatility patterns observed in a wide range of event-driven assets. For example, prior to major corporate restructurings, credit events, regulatory decisions, or the concentrated release of key operating information, markets are typically characterized by high uncertainty and elevated volatility; as outcomes become clearer and constraints tighten into a more stable geometry, volatility declines. In this sense, events are not exogenous shocks, but concentrated realizations of structural constraints. While traditional continuous-time models typically rely on exogenous stochastic volatility or jump processes to capture such “rise-then-fall” patterns, the financial relativity framework explains them as the endogenous adjustment of posterior geometry under the continuous arrival of constraints.

%
%
%
%
%
%
%

\section{Conclusion}

This paper develops a new interpretive framework for asset pricing from an information–geometric perspective, which we refer to as \emph{financial relativity}. Its central proposition is that terminal structural information determines how probability geometry is curved, while probability geometry determines how information is manifested in prices. Within this framework, the risk-neutral measure $Q$ is reinterpreted as the market probability geometry formed under structural constraints, and prices are understood as geometric projections of terminal payoffs onto the current information subspace. Price dynamics thus represent the progressive manifestation of structural information under geometric constraints.

Building on this idea, the paper makes three main contributions. First, through a structural analogy with physical relativity, it reorganizes the core objects of asset pricing theory—terminal payoffs, time–state–information filtration, probability measures, and price propagation—into a unified geometric framework. In this mapping, the flat prior corresponds to the unique inertial background, probability geometry to curvature, conditional expectations to local projections, and the unit conditional expectation of discounted returns to a local propagation invariant. Second, the paper introduces a prototype of the financial field equation, representing terminal structural information as the source of geometry, and provides explicit mechanisms for geometry formation in both closed systems (via discrete geometric potentials) and open systems (via the evolution of posterior densities). Third, it shows that price is not only a consequence of geometry but also the channel through which information becomes observable: in the discrete case, prices compress the state space through equivalence-class partitions and reveal terminal structure in a computable manner; in the continuous setting, the dynamic updating of posterior geometry induces a price geodesic equation, making volatility an endogenous function of posterior uncertainty.

The value of this framework lies primarily in shifting the explanatory focus of asset pricing theory. Rather than starting from probabilities and preferences, the present approach places terminal structure, probability geometry, and price manifestation along a single causal chain: structural constraints generate geometry, geometry organizes price propagation, and price paths transform structural information into observable quantities. This reorganization does not replace existing approaches—such as no-arbitrage pricing, stochastic discount factors, or behavioral models—but instead provides them with a deeper unifying background. In particular, when $Q$ is interpreted as probability geometry rather than a purely technical device, previously disconnected elements—such as the fundamental theorem of asset pricing, risk-neutral dynamics, information filtering, and volatility behavior—acquire a common theoretical foundation. Moreover, this framework offers a unified interpretation for phenomena that have often been studied separately: price changes can be understood as geometric reweighting induced by structural constraints rather than solely as preference shifts; information revelation emerges as a process of directional compression and progressive differentiation under geometry; and price volatility, risk premia, event-driven fluctuations, posterior contraction, and increasing state distinguishability can all be viewed as manifestations of the dynamic adjustment of probability geometry.

More importantly, financial relativity is not a closed model but an extensible research program. Its potential directions include at least four avenues. First, at the theoretical level, one may develop more general geometric operators and posterior dynamics, incorporating nonlinear field equations, information-geometric curvature, and high-dimensional state structures, thereby turning probability geometry into a fully analytical tool. Second, at the modeling level, the framework can be extended to multi-factor terminal structures, jump processes, complex constraint environments, and multiple-geometry selection problems in incomplete markets. Third, at the empirical level, one may estimate the evolution of posterior geometry using option-implied distributions, prediction markets, analyst expectations, and high-frequency data, and test structural relations between volatility and posterior uncertainty, as well as the origin of apparent risk premia across different reference frames. Fourth, from an informational perspective, one can construct operational measures of market efficiency by comparing the information content revealed by prices—quantified, for example, through conditional entropy or mutual information induced by price partitions—with the total information embedded in the filtration. Such measures would allow one to characterize the efficiency with which prices express terminal structure, and provide new tools for studying information asymmetry, delayed price adjustment, and market microstructure. In sum, financial relativity proposes a new organizing principle: to unify terminal structural information, probability measures, and price dynamics within a geometric framework, thereby offering a coherent, extensible, and empirically testable foundation for asset pricing theory.
\bibliographystyle{apalike}
\bibliography{financial_relativity_refs_revised}

\appendix

\section{Projection Proposition and a Three-Dimensional Illustration}\label{app:projection}

This appendix provides a rigorous proof of the statement “price is a projection” discussed in Section \ref{sec:analogy}, and illustrates it through a three-dimensional state space example. The projection interpretation of conditional expectation used in the main text relies on the following fundamental proposition.

\begin{proposition}[Projection Property of Conditional Expectation]\label{prop:projection-general}
Let $(\Omega,\mathcal F,Q)$ be a finite probability space, and define the inner product on $L^2(\Omega,\mathcal F,Q)$ as
\[
\langle Y,Z\rangle_Q=E^Q[YZ].
\]
Given any sub-$\sigma$-algebra $\mathcal G\subseteq\mathcal F$, for any $X\in L^2(\Omega,\mathcal F,Q)$, the random variable $E^Q[X\mid\mathcal G]$ is the unique element in the subspace $L^2(\mathcal G)$ satisfying
\[
\langle X-E^Q[X\mid\mathcal G],Y\rangle_Q=0,
\qquad \forall Y\in L^2(\mathcal G).
\]
That is, $E^Q[X\mid\mathcal G]$ is precisely the orthogonal projection of $X$ onto $L^2(\mathcal G)$.
\end{proposition}

\begin{proof}
By the definition of conditional expectation, $E^Q[X\mid\mathcal G]\in L^2(\mathcal G)$, and for any $A\in\mathcal G$,
\[
E^Q\big[(X-E^Q[X\mid\mathcal G])\mathbf 1_A\big]=0.
\]
Since any $Y\in L^2(\mathcal G)$ in a finite probability space can be expressed as a linear combination of $\mathcal G$-measurable indicator functions, it follows that
\[
E^Q\big[(X-E^Q[X\mid\mathcal G])Y\big]=0,
\qquad \forall Y\in L^2(\mathcal G),
\]
that is,
\[
\langle X-E^Q[X\mid\mathcal G],Y\rangle_Q=0,
\qquad \forall Y\in L^2(\mathcal G).
\]
This shows that $X-E^Q[X\mid\mathcal G]$ is orthogonal to the subspace $L^2(\mathcal G)$, and hence $E^Q[X\mid\mathcal G]$ is the orthogonal projection.

Uniqueness follows from the general property of projections in Hilbert spaces. If $Z\in L^2(\mathcal G)$ also satisfies
\[
\langle X-Z,Y\rangle_Q=0,
\qquad \forall Y\in L^2(\mathcal G),
\]
then letting $Y=Z-E^Q[X\mid\mathcal G]$ yields
\[
\|Z-E^Q[X\mid\mathcal G]\|_Q^2=0,
\]
so $Z=E^Q[X\mid\mathcal G]$, completing the proof.
\end{proof}

A direct implication of this proposition is that when the market can only distinguish equivalence classes under a given information partition, the price is precisely the optimal geometric representation of the terminal payoff in the subspace where states within each class must take equal values.

\subsection{A Three-Dimensional Illustration}

Consider the terminal state space
\[
\Omega=\{\omega_1,\omega_2,\omega_3\},
\]
with the flat prior
\[
Q(\omega_1)=Q(\omega_2)=Q(\omega_3)=\frac{1}{3}.
\]
Let the terminal payoff vector be
\[
X=(x_1,x_2,x_3).
\]

At the initial time, the market cannot distinguish any of the three states, so the admissible subspace consists only of constant vectors:
\[
\mathrm{span}\{(1,1,1)\}.
\]
Thus, the current price is
\[
\widetilde S_0=E^Q[X]=\dfrac{x_1+x_2+x_3}{3},
\]
and the projection vector is
\[
P_{\mathcal F_0}X
=
\Big(\dfrac{x_1+x_2+x_3}{3},\dfrac{x_1+x_2+x_3}{3},\dfrac{x_1+x_2+x_3}{3}\Big).
\]

In the next period, the market can distinguish $\omega_1$, but still cannot distinguish between $\omega_2$ and $\omega_3$. The information partition becomes
\[
\{\omega_1\},\qquad \{\omega_2,\omega_3\}.
\]
The corresponding admissible subspace is
\[
\mathrm{span}\{(1,0,0),(0,1,1)\}.
\]
Therefore,
\[
P_{\mathcal F_1}X=
\Big(x_1,\dfrac{x_2+x_3}{2},\dfrac{x_2+x_3}{2}\Big).
\]
This vector represents the price at the second period: the market can price $\omega_1$ separately, but assigns the same price to $\omega_2$ and $\omega_3$.

Figure \ref{fig:projection} illustrates this three-dimensional example geometrically. The terminal payoff vector $X$ is first projected onto the one-dimensional subspace $\mathrm{span}\{(1,1,1)\}$ at the initial time $t=0$. When the market enters the second period and can distinguish between $\omega_1$ and $\{\omega_2,\omega_3\}$, the projection subspace expands to $\mathrm{span}\{(1,0,0),(0,1,1)\}$, and the price representation is accordingly refined. However, prices are always represented as vectors of the same dimension as the terminal payoff. This figure provides an intuitive visualization of the central statement in the main text: “price is projection.”

\begin{figure}[H]
\centering
  \includegraphics[width=0.8\textwidth]{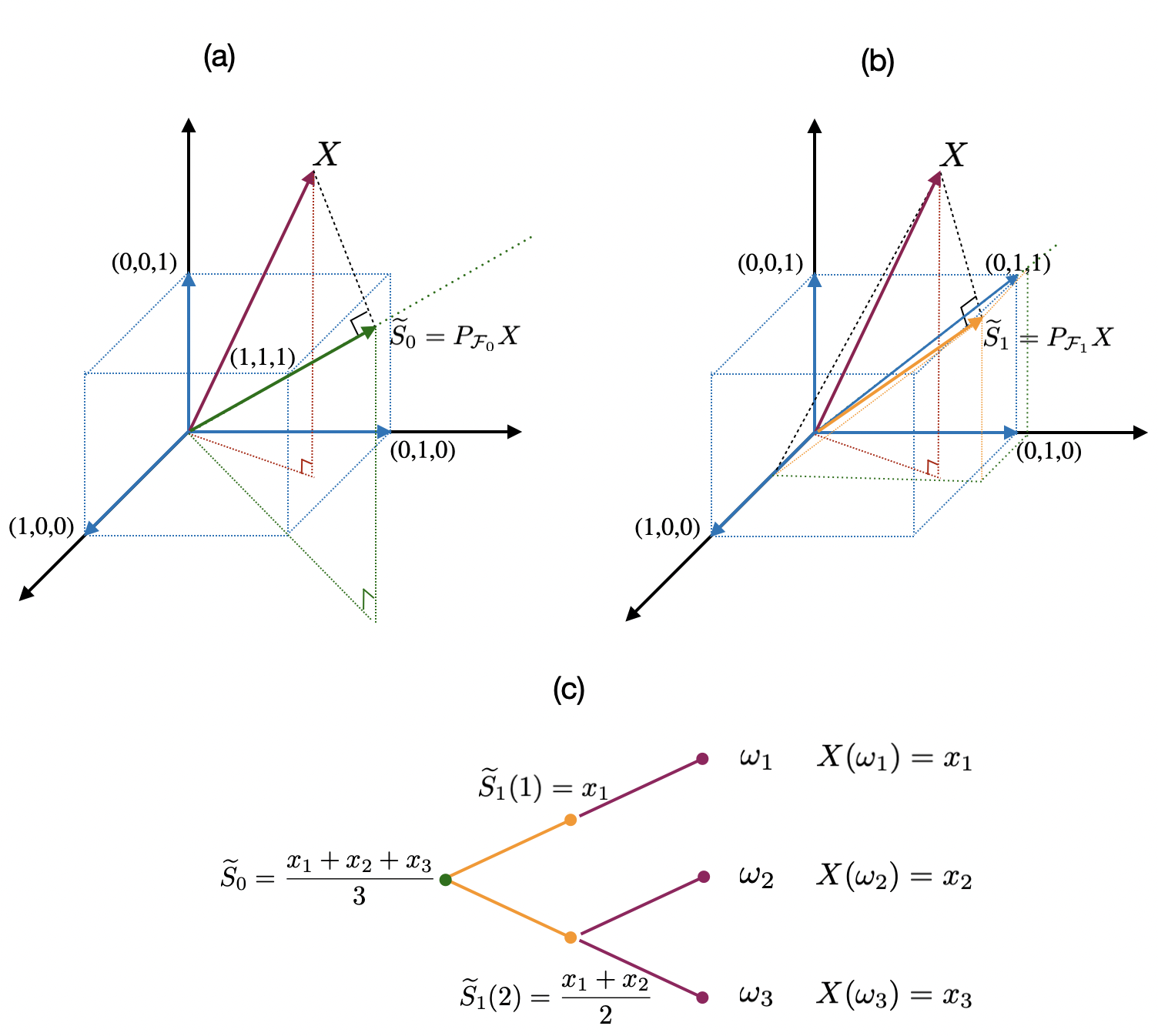}
  \caption{Schematic illustration of the price process as the projection of the terminal payoff vector onto information subspaces. (a) At the initial time $t=0$, the discounted price $\widetilde S_0$ is the orthogonal projection of $X$ onto the one-dimensional subspace $\mathrm{span}\{(1,1,1)\}$. (b) At the intermediate time $t=1$, the discounted price $\widetilde S_1$ is the orthogonal projection of $X$ onto the two-dimensional subspace $\mathrm{span}\{(1,0,0),(0,1,1)\}$ associated with the filtration $\mathcal F_1$. (c) From the perspective of a standard tree representation, the price process is expressed as a random variable whose values are given by conditional expectations, corresponding to the numerical realization of these projections along different information branches.}
  \label{fig:projection}
\end{figure}
\section{An Entropic Proof of the Information Conservation Law}\label{app:info-conservation}

This appendix provides a rigorous formulation and proof of the information conservation law used in Section 3 of the main text.

Let the terminal state space be a finite set
\[
\Omega_N=\{\omega_1,\dots,\omega_m\},
\]
equipped with a probability measure $P$ such that $P(\omega)>0$ for all $\omega\in\Omega_N$. At period $n$, the information structure is induced by a finite partition of the terminal state space:
\[
\Pi_n=\{A_1^{(n)},A_2^{(n)},\dots,A_{k_n}^{(n)}\},
\]
where each $A_i^{(n)}\subset \Omega_N$ is disjoint, and
\[
\bigcup_{i=1}^{k_n} A_i^{(n)}=\Omega_N.
\]
The corresponding $\sigma$-algebra at time $n$ is
\[
\mathcal F_n=\sigma(\Pi_n).
\]
Denote the probability of each branch by
\[
p_i^{(n)}:=P(A_i^{(n)}),\qquad i=1,\dots,k_n.
\]

\begin{definition}[Total terminal entropy]
The total entropy of the terminal state space is defined as
\[
H_N:=-\sum_{\omega\in\Omega_N} P(\omega)\ln P(\omega).
\]
\end{definition}

\begin{definition}[Branch entropy at time $n$]
The entropy revealed by the partition $\Pi_n$ at time $n$ is defined as
\[
H_n:=-\sum_{i=1}^{k_n} p_i^{(n)}\ln p_i^{(n)}.
\]
\end{definition}

\begin{definition}[Residual terminal entropy after time $n$]
Under the information structure at time $n$, conditional on being in branch $A_i^{(n)}$, there remains uncertainty about the terminal state within that branch. The corresponding conditional entropy is
\[
H\!\left(\Omega_N\mid A_i^{(n)}\right)
:=
-\sum_{\omega\in A_i^{(n)}} P(\omega\mid A_i^{(n)})\ln P(\omega\mid A_i^{(n)}).
\]
The expected residual terminal entropy after time $n$ is therefore defined as
\[
H_{N\mid n}
:=
\sum_{i=1}^{k_n} p_i^{(n)}\,H\!\left(\Omega_N\mid A_i^{(n)}\right).
\]
\end{definition}

\begin{definition}[Information revealed at time $n$]
The amount of information revealed by the information structure $\mathcal F_n$ is defined as
\[
I(\mathcal F_n):=H_N-H_{N\mid n}.
\]
\end{definition}

\begin{theorem}[Information Conservation Law]\label{thm:info-conservation}
In the above finite-state multi-period setting, for any information structure $\mathcal F_n$ induced by a partition $\Pi_n$ of the terminal state space, we have
\[
H_N = H_n + H_{N\mid n}.
\]
Equivalently,
\[
I(\mathcal F_n)=H_n.
\]
Thus, total terminal uncertainty decomposes exactly into the uncertainty revealed by the branching structure at time $n$ and the residual uncertainty within each branch. If no new external structural constraints are introduced, this decomposition implies that total information remains conserved over time.
\end{theorem}

\begin{proof}
By the definition of conditional entropy,
\[
H_{N\mid n}
=
\sum_{i=1}^{k_n} p_i^{(n)}
\left(
-\sum_{\omega\in A_i^{(n)}} P(\omega\mid A_i^{(n)})\ln P(\omega\mid A_i^{(n)})
\right).
\]
Multiplying the outer probability into the inner sum yields
\[
H_{N\mid n}
=
-\sum_{i=1}^{k_n}\sum_{\omega\in A_i^{(n)}}
p_i^{(n)}P(\omega\mid A_i^{(n)})\ln P(\omega\mid A_i^{(n)}).
\]
Since
\[
p_i^{(n)}P(\omega\mid A_i^{(n)})=P(\omega),
\qquad \omega\in A_i^{(n)},
\]
we obtain
\[
H_{N\mid n}
=
-\sum_{i=1}^{k_n}\sum_{\omega\in A_i^{(n)}}
P(\omega)\ln P(\omega\mid A_i^{(n)}).
\]
Moreover,
\[
P(\omega\mid A_i^{(n)})=\frac{P(\omega)}{p_i^{(n)}},
\qquad \omega\in A_i^{(n)},
\]
so that
\[
\ln P(\omega\mid A_i^{(n)})
=
\ln P(\omega)-\ln p_i^{(n)}.
\]
Substituting back, we obtain
\[
H_{N\mid n}
=
-\sum_{i=1}^{k_n}\sum_{\omega\in A_i^{(n)}}
P(\omega)\bigl[\ln P(\omega)-\ln p_i^{(n)}\bigr].
\]
Separating the two terms gives
\[
H_{N\mid n}
=
-\sum_{\omega\in\Omega_N}P(\omega)\ln P(\omega)
+
\sum_{i=1}^{k_n}\sum_{\omega\in A_i^{(n)}}P(\omega)\ln p_i^{(n)}.
\]
The first term is precisely $H_N$. In the second term, $\ln p_i^{(n)}$ is constant within each branch, hence
\[
\sum_{\omega\in A_i^{(n)}}P(\omega)\ln p_i^{(n)}
=
p_i^{(n)}\ln p_i^{(n)}.
\]
Therefore,
\[
H_{N\mid n}
=
H_N+\sum_{i=1}^{k_n}p_i^{(n)}\ln p_i^{(n)}.
\]
By the definition of $H_n$,
\[
H_n=-\sum_{i=1}^{k_n}p_i^{(n)}\ln p_i^{(n)},
\]
so that
\[
H_{N\mid n}=H_N-H_n.
\]
Rearranging yields
\[
H_N=H_n+H_{N\mid n}.
\]
Finally, by the definition of revealed information,
\[
I(\mathcal F_n)=H_N-H_{N\mid n},
\]
we immediately obtain
\[
I(\mathcal F_n)=H_n.
\]
This completes the proof.
\end{proof}

\section{Solution of the Discrete Field Equation under a Block-Symmetric Source}\label{app:block-symmetric-field}

This appendix retains only the technical results that are genuinely needed for the discrete example in Section 4. In particular, it explains why, under a block-symmetric source, the geometric potential can be written as two blockwise constants, and why it must satisfy
\[
3\phi_A+5\phi_B=0.
\]
The purpose here is not to re-prove the general proposition that “conditional expectation equals projection,” but to provide a closed-form geometric solution for the eight-state, three-period example in the main text.

\subsection{Setup}

Let the terminal state space be
\[
\Omega=\{\omega_1,\omega_2,\dots,w_8\},
\]
and let the first-period information structure partition it into two branches:
\[
A=\{\omega_1,\omega_2,\omega_3\},\qquad
B=\{w_4,w_5,w_6,w_7,w_8\}.
\]
The structureless prior is given by
\[
P(w_i)=\frac18,\qquad i=1,\dots,8.
\]

Assume that the terminal structural source takes the block-symmetric form
\[
\rho(w)=
\begin{cases}
\rho_A,& w\in A,\\[4pt]
\rho_B,& w\in B.
\end{cases}
\]
The balancing condition adopted in the main text is
\[
3\rho_A+5\rho_B=0.
\]
This condition means that the source term merely rearranges probability geometry as a whole, without introducing any additional shift in total mass.

The geometric potential function $\phi$ is defined through the logarithmic geometric relation
\[
Q(w)=\frac{P(w)e^{\phi(w)}}{\sum_{u\in\Omega}P(u)e^{\phi(u)}}.
\]
To determine $\phi$ from the source term, the main text adopts the prototype discrete field equation
\[
L\phi=\kappa\rho,
\]
where $\kappa>0$ is the coupling strength and $L$ is a discrete Laplacian operator on the state space.

\subsection{Form of the Block-Symmetric Solution}

Since the source term is constant within each block, and the prior $P$ treats states symmetrically within each block, the field equation and the prior together induce permutation invariance within blocks. Accordingly, one may seek the geometric potential within the class of block-symmetric solutions, namely,
\[
\phi(w)=
\begin{cases}
\phi_A,& w\in A,\\[4pt]
\phi_B,& w\in B.
\end{cases}
\]
This is not an additional economic assumption, but the most natural solution class under a block-symmetric source and a block-symmetric background. Allowing asymmetric terms within each block would correspond to a finer layer of informational structure, which lies beyond the purpose of the example in the main text.

\subsection{Why \texorpdfstring{$3\phi_A+5\phi_B=0$}{3phiA+5phiB=0} Must Hold}

The key point is that the geometry $Q$ generated by $\phi$ is invariant under the addition of a constant to $\phi$. Indeed, for any constant $c$,
\[
\frac{P(w)e^{\phi(w)+c}}{\sum_{u}P(u)e^{\phi(u)+c}}
=
\frac{P(w)e^{\phi(w)}}{\sum_{u}P(u)e^{\phi(u)}}.
\]
Thus, $\phi$ is not unique; it is meaningful only up to an additive constant. To select a unique representative from this equivalence class, one must impose a \textbf{gauge condition}.

In the present block-symmetric example, the most natural gauge condition is that the geometric potential has zero mean with respect to the flat prior $P$, namely,
\[
E^{P}[\phi]=0.
\]
Since $P(w_i)=1/8$, this means
\[
\frac18\sum_{i=1}^8 \phi(w_i)=0.
\]
Under a block-symmetric solution,
\[
\sum_{i=1}^8 \phi(w_i)=3\phi_A+5\phi_B,
\]
and therefore
\[
3\phi_A+5\phi_B=0.
\]
This is the origin of the relation used in the main text. It does not follow directly from the balancing condition on the source term, but rather from the \textbf{gauge choice for the geometric potential}; the balancing condition
\[
3\rho_A+5\rho_B=0
\]
ensures that the block-symmetric field equation is compatible with this gauge.

\subsection{Solving the Block-Symmetric Field Equation}

To obtain an explicit solution, the main text adopts the simplest cross-block Laplacian structure: each $w_i\in A$ interacts only with states in block $B$, and each $w_j\in B$ interacts only with states in block $A$. Setting each cross-block edge weight equal to $1$, for any $w_i\in A$,
\[
(L\phi)(w_i)=\sum_{u\in B}\bigl(\phi_A-\phi_B\bigr)=5(\phi_A-\phi_B),
\]
while for any $w_j\in B$,
\[
(L\phi)(w_j)=\sum_{u\in A}\bigl(\phi_B-\phi_A\bigr)=3(\phi_B-\phi_A).
\]
Hence the field equation
\[
L\phi=\kappa\rho
\]
reduces to
\[
5(\phi_A-\phi_B)=\kappa\rho_A,
\qquad
3(\phi_B-\phi_A)=\kappa\rho_B.
\]
These two equations are compatible if and only if
\[
3\rho_A+5\rho_B=0.
\]
This is precisely the balancing condition adopted in the main text.

Combining this with the gauge condition
\[
3\phi_A+5\phi_B=0,
\]
one uniquely obtains
\[
\phi_A=\frac58\,(\phi_A-\phi_B),
\qquad
\phi_B=-\frac38\,(\phi_A-\phi_B).
\]
From
\[
5(\phi_A-\phi_B)=\kappa\rho_A,
\]
it follows that
\[
\phi_A-\phi_B=\frac{\kappa\rho_A}{5}.
\]
Therefore,
\[
\phi_A=\frac{\kappa\rho_A}{8},\qquad
\phi_B=-\frac{3\kappa\rho_A}{40}.
\]
If, as in the main text, one sets
\[
\rho_A=5,\qquad \rho_B=-3,
\]
this further reduces to
\[
\phi_A=\frac58\kappa,\qquad \phi_B=-\frac38\kappa.
\]
These are exactly the expressions used in the main text.

\section{Proof of Theorem 5.1} \label{app:theoremkkk}
\begin{proof}
By equation \ref{eqn:iienagg},
\[
q_t(x)=\frac{\widetilde q_t(x)}{Z_t},
\qquad
\widetilde q_t(x):=P(x)\exp\!\left(\sigma x\,\xi_t-\frac12\sigma^2x^2t\right).
\]
We first differentiate the unnormalized density $\widetilde q_t(x)$. Define
\[
M_t(x):=\sigma x\,\xi_t-\frac12\sigma^2x^2t.
\]
By It\^o's formula,
\[
dM_t(x)=\sigma x\,d\xi_t-\frac12\sigma^2x^2dt,
\qquad
(dM_t(x))^2=\sigma^2x^2(d\xi_t)^2=\sigma^2x^2dt.
\]
Therefore,
\[
d e^{M_t(x)}
=
e^{M_t(x)}
\left(
dM_t(x)+\frac12(dM_t(x))^2
\right)
=
e^{M_t(x)}\sigma x\,d\xi_t.
\]
Hence
\[
d\widetilde q_t(x)=\sigma x\,\widetilde q_t(x)\,d\xi_t.
\]

Now consider the normalizing factor
\[
Z_t=\int_{\mathbb R}\widetilde q_t(y)\,dy.
\]
Integrating the above expression yields
\[
dZ_t
=
\int_{\mathbb R}d\widetilde q_t(y)\,dy
=
\sigma\int_{\mathbb R}y\,\widetilde q_t(y)\,dy\,d\xi_t
=
\sigma m_t Z_t\,d\xi_t.
\]
Thus,
\[
\frac{dZ_t}{Z_t}=\sigma m_t\,d\xi_t.
\]

Now apply the product rule and It\^o's formula to
\[
q_t(x)=\widetilde q_t(x)Z_t^{-1}.
\]
First,
\[
d(Z_t^{-1})
=
-Z_t^{-2}dZ_t+Z_t^{-3}(dZ_t)^2.
\]
Since
\[
dZ_t=\sigma m_t Z_t\,d\xi_t,
\]
we obtain
\[
d(Z_t^{-1})
=
-\sigma m_t Z_t^{-1}d\xi_t+\sigma^2m_t^2 Z_t^{-1}dt.
\]
At the same time,
\[
d\widetilde q_t(x)\,d(Z_t^{-1})
=
\bigl(\sigma x\widetilde q_t(x)d\xi_t\bigr)
\bigl(-\sigma m_t Z_t^{-1}d\xi_t\bigr)
=
-\sigma^2xm_t q_t(x)dt.
\]
Therefore,
\[
dq_t(x)
=
Z_t^{-1}d\widetilde q_t(x)
+
\widetilde q_t(x)d(Z_t^{-1})
+
d\widetilde q_t(x)d(Z_t^{-1}).
\]
Substituting the above expressions gives
\[
dq_t(x)
=
\sigma x q_t(x)d\xi_t
-\sigma m_t q_t(x)d\xi_t
+\sigma^2m_t^2 q_t(x)dt
-\sigma^2xm_t q_t(x)dt.
\]
Rearranging,
\[
dq_t(x)
=
\sigma(x-m_t)q_t(x)d\xi_t
-\sigma^2m_t(x-m_t)q_t(x)dt.
\]
Using
\[
d\xi_t=\sigma m_tdt+dW_t^Q,
\]
we obtain
\[
dq_t(x)
=
\sigma(x-m_t)q_t(x)(\sigma m_tdt+dW_t^Q)
-\sigma^2m_t(x-m_t)q_t(x)dt.
\]
The $dt$ terms cancel exactly, leaving
\[
dq_t(x)=\sigma(x-m_t)q_t(x)dW_t^Q.
\]
This proves the first statement.

Now apply It\^o's formula to $\log q_t(x)$:
\[
d\log q_t(x)
=
\frac{dq_t(x)}{q_t(x)}
-\frac12\frac{(dq_t(x))^2}{q_t(x)^2}.
\]
Since
\[
dq_t(x)=\sigma(x-m_t)q_t(x)dW_t^Q,
\]
it follows that
\[
\frac{dq_t(x)}{q_t(x)}=\sigma(x-m_t)dW_t^Q,
\qquad
\frac{(dq_t(x))^2}{q_t(x)^2}=\sigma^2(x-m_t)^2dt.
\]
Therefore,
\[
d\log q_t(x)
=
\sigma(x-m_t)dW_t^Q
-\frac12\sigma^2(x-m_t)^2dt.
\]
The theorem is proved.
\end{proof}

\end{document}